\documentclass[conference]{IEEEtran}

\usepackage{geometry}
\usepackage[utf8]{inputenc} 
\usepackage[T1]{fontenc}
\usepackage{url}
\usepackage{ifthen}
\usepackage{cite}
\usepackage{comment}
\usepackage[cmex10]{amsmath} 
\usepackage{amsthm}
\newtheorem{theorem}{Theorem} 
\newtheorem{example}{Example}
\usepackage{array}
\usepackage{graphicx}
\usepackage{mathtools}
\newtheorem{definition}{Definition}
\newtheorem{lemma}{Lemma}
\newtheorem{remark}{Remark}

\begin{document}
\title{Construction of CCC and ZCCS Through Additive Characters Over Galois Field} 


 \author{
   \IEEEauthorblockN{Gobinda Ghosh}
   \IEEEauthorblockA{School of Technology\\
                     Woxsen University, Hyderabad, Telengana\\
                     Email: sagarghosh798@gmail.com}
   \and
   \IEEEauthorblockN{Sachin Pathak}
   \IEEEauthorblockA{Department of Mathematics and Basic Science\\ 
                     NIIT University, Rajasthan, India
 \\
                     Email: sachiniitk93@gmail.com}
 }

\maketitle


\begin{abstract}

With the rapid progression in wireless communication technologies, especially in multicarrier code-division multiple access (MC-CDMA), there is a need of advanced code construction methods. Traditional approaches, mainly based on generalized Boolean functions, have limitations in code length versatility. This paper introduces a novel approach to constructing complete complementary codes (CCC) and Z-complementary code sets (ZCCS), for reducing interference in MC-CDMA systems. The proposed construction, distinct from Boolean function-based approaches, employs additive characters over Galois fields GF($p^{r}$), where $p$ is prime and $r$ is a positive integer. First, we develop CCCs with lengths of $p^{r}$, which are then extended to construct ZCCS with both unreported lengths and sizes of $np^{r}$, where $n$ are arbitrary positive integers. The versatility of this method is further highlighted as it includes the lengths of ZCCS reported in prior studies as special cases, underscoring the method's comprehensive nature and superiority.

\end{abstract}

\section{Introduction}

The advancement of wireless communication technologies, particularly in the field of multicarrier code-division multiple access (MC-CDMA), is a critical area of research in today's rapidly evolving digital era. Traditional code construction methods, primarily based on generalized Boolean functions, have been instrumental yet exhibit certain constraints, particularly in terms of code length versatility. This paper introduces a new approach to the development of complete complementary codes (CCC) and  Z-complementary code sets (ZCCS), which are essential for reducing interference in wireless communication systems.\\
Fan \textit{et al}. \cite{fan2007z} significantly contributed to the development of ZCCS for MC-CDMA systems, a pivotal advancement highlighted in their study \cite{fan2007z}. These ZCCSs, distinguished by their zero correlation zone (ZCZ), play a crucial role in enhancing the performance of MC-CDMA, particularly in quasi-synchronous channels by enabling interference-free operations without the need for power adjustment \cite{chen2007next}. Research in the area of ZCCS has spanned various construction methodologies, both direct \cite{wu2018optimal,sarkar2020direct,wu2020z,shen2022new,ghosh2022direct,sarkar2020construction,xie2021constructions} and indirect \cite{adhikary2019new,yu2022new,b}.

Existing methods in the current literature are limited in terms of the lengths and sizes available for optimal ZCCS. This constraint has motivated our research to explore new methodologies for constructing ZCCS with novel lengths and sizes previously unattained in the existing literature. Our novel approach diverges from traditional boolean function-based methodologies by employing additive characters over finite Galois fields. This innovative approach enables the construction of CCC with lengths of $p^{r}$, and ZCCS with both unreported lengths and sizes of $np^{r}$, where $n,r$ are arbitrary positive integers. Furthermore, the versatility of this method is exemplified by the fact that the lengths of ZCCS can be included in those reported in important previous studies, such as those in \cite{xie2021constructions,ghosh2022direct,ghosh2023construction} as special cases within our broader framework.
\par The structure of the rest of this work is organized as follows: Section II delves into essential preliminary concepts. The methodology for constructing CCC is detailed in Section III. Section IV is dedicated to the proposed methodology for the development of ZCCS. The paper concludes with key findings and observations in Section V.

\section{Preliminary}\label{Sect 2}
In this section, we lay the foundational concepts and essential lemmas that will be integral to our proposed construction approach.
\begin{definition}
     Let
$\mathbf{a}=(a_{0},\hdots,a_{l-1})$ and $\mathbf{b}=(b_{0},\hdots ,b_{l-1})$ be two complex valued sequences.
We define the aperiodic cross-correlation function (ACCF) between $\mathbf{a}$ and $\mathbf{b}$ as
\begin{equation}\label{equ:cross}
\Phi(\mathbf{a}, \mathbf{b})(\tau)=\begin{cases}
\sum_{k=0}^{l-1-\tau}a_{k}\overline b_{k+\tau}, & 0 \leq \tau < l, \\
\sum_{k=0}^{l+\tau -1}a_{k-\tau}\overline b_{k}, & -l < \tau < 0, \\
0, & \text{otherwise},
\end{cases}
\end{equation}
where $\overline b_{k}$ denotes the complex conjugate of $b_{k}$.
If $\mathbf{a}$ is equal to $\mathbf{b}$, then the resulting function is termed the aperiodic auto-correlation function (AACF) of $\mathbf{a}$, denoted as $\Phi(\mathbf{a})$.
\end{definition}
\begin{definition}
    Consider the collection of sequences \(\mathbf{D}_{i}=\{\mathbf{a}_{k}^{i}:0 \leq k < m\}\) and \(\mathbf{D}_{j}=\{\mathbf{a}_{k}^{j}:0 \leq k < m\}\), each comprising \(m\) sequences. For each sequence \(\mathbf{a}_{k}^{i}\), it is represented as \(({a}_{k,0}^{i},\ldots,{a}_{k,l-1}^{i})\). The aperiodic cross-correlation sum (ACCS) for these two sets of sequences, \(\mathbf{D}_{i}\) and \(\mathbf{D}_{j}\), is defined by:
$$\Phi(\mathbf{D}_{i},\mathbf{D}_{j})(\tau)=\sum_{k=0}^{m-1}\Phi(\mathbf{a}_{k}^{i},\mathbf{a}_{k}^{j})(\tau).$$
\end{definition}

\begin{definition}
Consider \(\mathbf{D}=\{\mathbf{D}_{0},\ldots,\mathbf{D}_{s-1}\}\) as a collection of \(p\) such sequence sets. The code set \(\mathbf{D}\) is defined as a \((s,m,l,z)\)-ZCCS (refer to \cite{ghosh2022direct}) if it satisfies the following criteria:
\begin{eqnarray}
\Phi(\mathbf{D}_{\nu_1},\mathbf{D}_{\nu_2})(\tau)
=\begin{cases}
ml, & \tau=0, \nu_1=\nu_2,\\
0, & 0<|\tau|<z, \nu_1=\nu_2,\\
0, & |\tau|< z, \nu_1\neq \nu_2,
\end{cases}
\end{eqnarray}
when $z=l,s=n$ it refears as CCC with parameter $(p,p,l)-CCC.$
\end{definition}
\begin{lemma}\cite{liu2011correlation}
For a \((s,m,l,z)\)-ZCCS configuration, it is established that \(s \leq m \left\lfloor \frac{l}{z} \right\rfloor\), where \(s\), \(m\), \(l\), and \(z\) represent the number of users, the number of sub-carriers, the code length, and the ZCZ width, respectively. The ZCCS is considered optimal when
\begin{equation}
     s = m \left\lfloor \frac{l}{z} \right\rfloor,
\end{equation}
where $\left\lfloor . \right\rfloor$ denotes floor function.
\end{lemma}
\subsection{Characters on Finite Fields}
In this section, we discuss the concept of characters on finite fields, as detailed in \cite{lidl1997finite,ireland1990classical}.

    Let  $p$ be a prime number and $q=p^{r}$ for some positive integer $r$ and $GF(q)$ be a finite Galois field of order $q=p^r$.  The absolute trace from the field $GF(q)$ to its prime subfield $GF(p)$ is given by:
$$
\begin{aligned}
\operatorname{Tr}: GF(q) & \rightarrow GF(p) \\
c & \mapsto c + c^p + c^{p^2} + \cdots + c^{p^{r-1}} .
\end{aligned}
$$
A  additive character $\chi_b$ where $b\in GF(q) $ is a homomorphism from $GF(q)$ to the field of complex number $\mathbf{C}$ defined by
$$
\chi_b(c) = e^{\frac{2 \pi i}{p} \operatorname{Tr}(bc)}, \text{ for all } c \in GF(q).
$$
When $b=1$ it is called a canonical additive character. Furthermore, any additive character over \( GF(q) \) can be derived as:
$$
\chi_b(c) := \chi_1(bc), \text{ for } b, c \in GF(q).
$$
    Let \( \chi_a \) and \( \chi_b \) be additive characters of \( GF(q) \). Then the orthogonality relations for these characters are given by:
\begin{equation*}
\sum_{c \in GF(q)} \chi_a(c) \overline{\chi_b(c)} = 
\begin{cases}
q, & \text{if } a = b \\
0, & \text{if } a \neq b .
\end{cases}
\end{equation*}
\section{Constructions of CCC}

Let $q=p^{r}$ and $\alpha$ be a primitive element of $GF^{*}(q)$, where $GF^{*}(q)$ is the nonzero elements of $GF(q)$ . For $0\leq i\leq q-1$ we define
\begin{equation}
    a(i) = 
\begin{cases} 
0 & \text{if } i = 0, \\
\alpha^{i-1} & \text{for } 0 < i \leq p^{r}-1.
\end{cases}
\end{equation}

For any integer $0\leq k \leq q-1$ we define the vector representation as $\mathbf{k}=(k_{1},k_{2},\ldots,k_{r})$, where  $$k=\sum_{i=1}^{r}k_{i}p^{i-1},$$ $0\leq k_{i}\leq p-1$. Let $0\leq k,l \leq q-1$ and  $S_{k,l}$ be the function defined from $\{0,1,\ldots,q-1\}$ to the field of complex number $\mathbf{C}$ defined by 
\begin{equation}
\label{Krishna}
S_{k,l}(i)=\omega_{p}^{(\mathbf{k}.\mathbf{i})+Tr\left(a(i)a(l)\right)},  
\end{equation}
where $\mathbf{k}$ and $\mathbf{i}$ are the vectors corresponding to $k$ and $i$ respectively and $\cdot$ denote usual dot product.
We define the sequence corresponding to $S_{k,l}$ as 
$$\psi(S_{k,l})=\left(S_{k,l}(0),S_{k,l}(1),\ldots,S_{k,l}(q-1)\right).$$
\begin{theorem}
Let $q=p^{r}$ and $\alpha$ be the primitive element in $GF^{*}(q)$ and $S_{k,l}$ be the function defined in (\ref{Krishna}), now consider the ordered set of sequences or code as 
\begin{equation}
    \psi(S_{k})=\left\{\psi(S_{k,l}):0\leq l \leq q-1  \right\}.
\end{equation}
Define the set of codes as 
\begin{equation}
    S=\{ \psi(S_{k}):0\leq k\leq q-1\}.
\end{equation}
Then the code set $S$ forms CCC with parameter $(p^{r},p^{r},p^{r}).$
\end{theorem}
\begin{proof}
    Let $\psi(S_{k_{1}})$ and  and $\psi(S_{k_{2}})$ be two codes in $S$. The ACCF between $\psi(S_{k_{1}})$ and $\psi(S_{k_{2}})$ is defined as
    \begin{equation}
    \begin{split}
        &\Phi\left(\psi(S_{k_{1}}),\psi(S_{k_{1}})\right)(\tau)\\
        &=\sum_{l=0}^{p^{r}-1}\sum_{i=0}^{p^{r}-1-\tau} \omega_{p}^{(\mathbf{k_{1}}.\mathbf{i})- (\mathbf{k_{2}}.\mathbf{(i+\tau)})} \omega_{p}^{Tr\left(a(i)a(l)\right)-Tr\left(a(i+\tau)a(l)\right)}\\
        &=\sum_{i=0}^{p^{r}-1-\tau} \left(\sum_{l=0}^{p^{r}-1}\chi_{a(i)}(a(l))\overline{\chi_{a(i+\tau)}(a(l))}
\right)\omega_{p}^{(\mathbf{k_{1}}.\mathbf{i})- (\mathbf{k_{2}}.\mathbf{(i+\tau)})}
    \end{split}
    \end{equation}
\begin{enumerate}
    \item Case 1 $(k_{1}\neq k_{2})$: In this case we have two subcases
    \begin{enumerate}
        \item Subcase $( \tau \neq 0 ):$ In this case we have
    $$\sum_{l=0}^{p^{r}-1}\chi_{a(i)}(a(l))\overline{\chi_{a(i+\tau)}(a(l))}=0.$$
    \item Subcase $( \tau = 0 ):$ In this case we have 
    $$\sum_{i=0}^{p^{r}-1}\omega_{p}^{(\mathbf{k_{1}}.\mathbf{i})- (\mathbf{k_{2}}.\mathbf{i})}=0.$$
    \end{enumerate}
\item Case 2 $(k_{1}= k_{2})$: In this case we have two following subcases
\begin{enumerate}
    \item Subcase $( \tau \neq 0 ):$ In this case we have
     $$\sum_{l=0}^{p^{r}-1}\chi_{a(i)}(a(l))\overline{\chi_{a(i+\tau)}(a(l))}=0.$$
     \item Subcase $( \tau = 0 ):$ In this case we have
     $$\Phi(\psi(S_{k_{1}},\psi(S_{k_{1}})(0)=p^{2r}.$$
\end{enumerate}
Hence combining case 1 and case 2 we conclude that the code set $S$ forms CCC with parameter $(p^{r},p^{r},p^{r})$.    
\end{enumerate}
\end{proof}
\begin{example}
    Let $r=2,p=3$ and consider the Galois field $F_{3}[x]/<x^{2}+x+2>$=GF($3^{2}$). Let $\alpha$ be a primitive element with $\alpha^{2}+\alpha+2=0$. Consider the function $S_{k,l}:\{0,1,\ldots,8\}\rightarrow $
    $\mathbf{C}$ given by 
    $$S_{k,l}(i)=\omega_{3}^{(\mathbf{k}.\mathbf{i})+Tr\big((a(i)a(l)\big)},$$ where $0\leq i,k,l \leq 8$. Therefore we have 
$$\psi(S_{k})=\{\psi(S_{k,l}):0\leq l\leq 8\},$$
where $$\psi(S_{k,l})=(S_{k,l}(0),S_{k,l}(1),\ldots,S_{k,l}(8)),$$ and $$S_{k,l}(i)=\omega_{3}^{\mathbf{k}.\mathbf{i}}\chi_{a(i)}\left(a(l)\right).$$Hence the set
\begin{equation}
    S=\{ \psi(S_{k}):0\leq k\leq 8\},
\end{equation}
forms CCC with parameter $(9,9,9)$.  Fig. \ref{fig3} and Fig. \ref{fig4} show the auto and cross-correlation results of any two codes from $S$.
\end{example}
\begin{figure}
\centering
\includegraphics[width=0.5\textwidth]{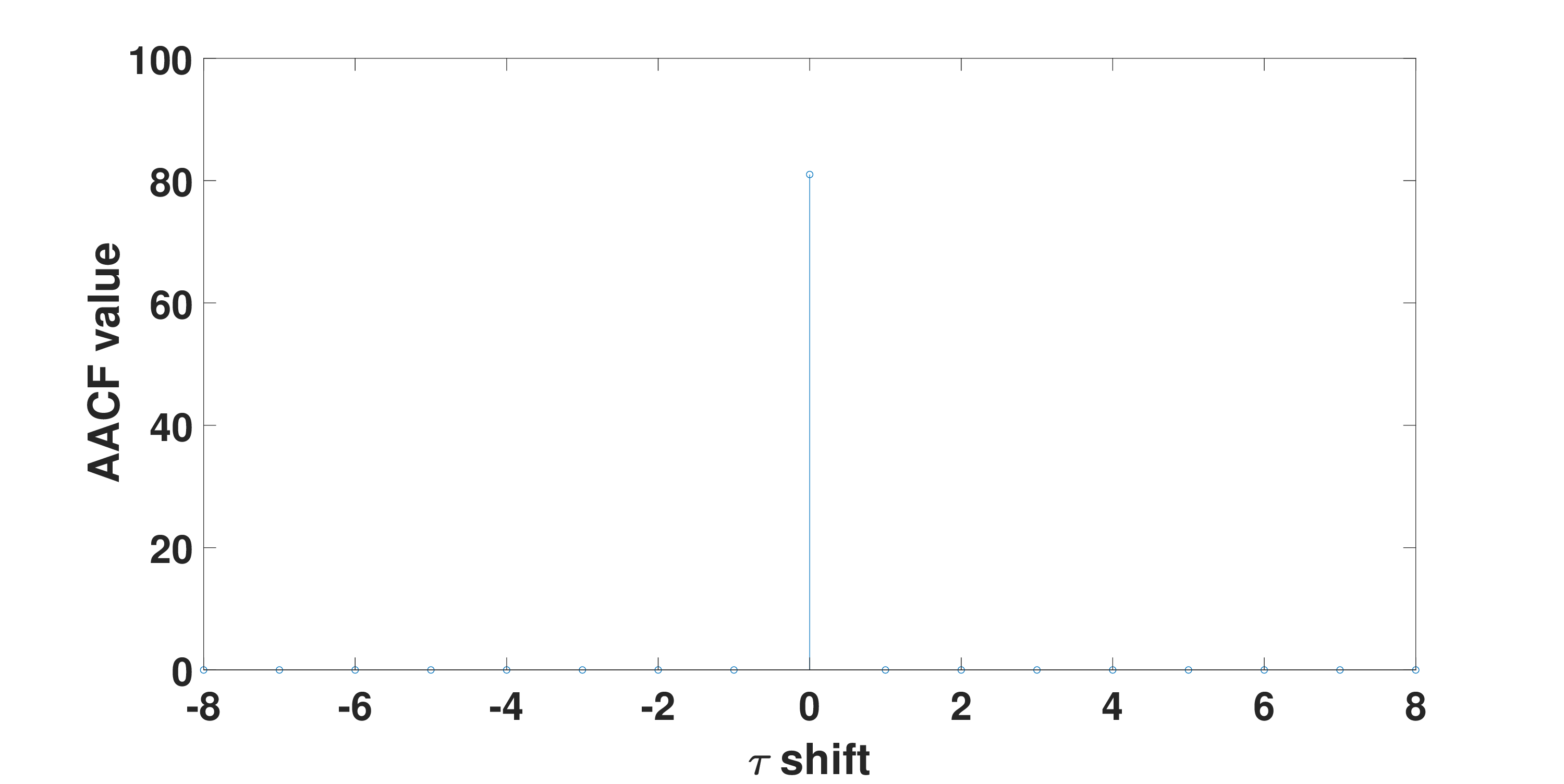}
\caption{Auto-correlation result of any set of array from $\mathcal{S}$}\label{fig3}
\end{figure}
\begin{figure}
\centering
\includegraphics[width=0.5\textwidth]{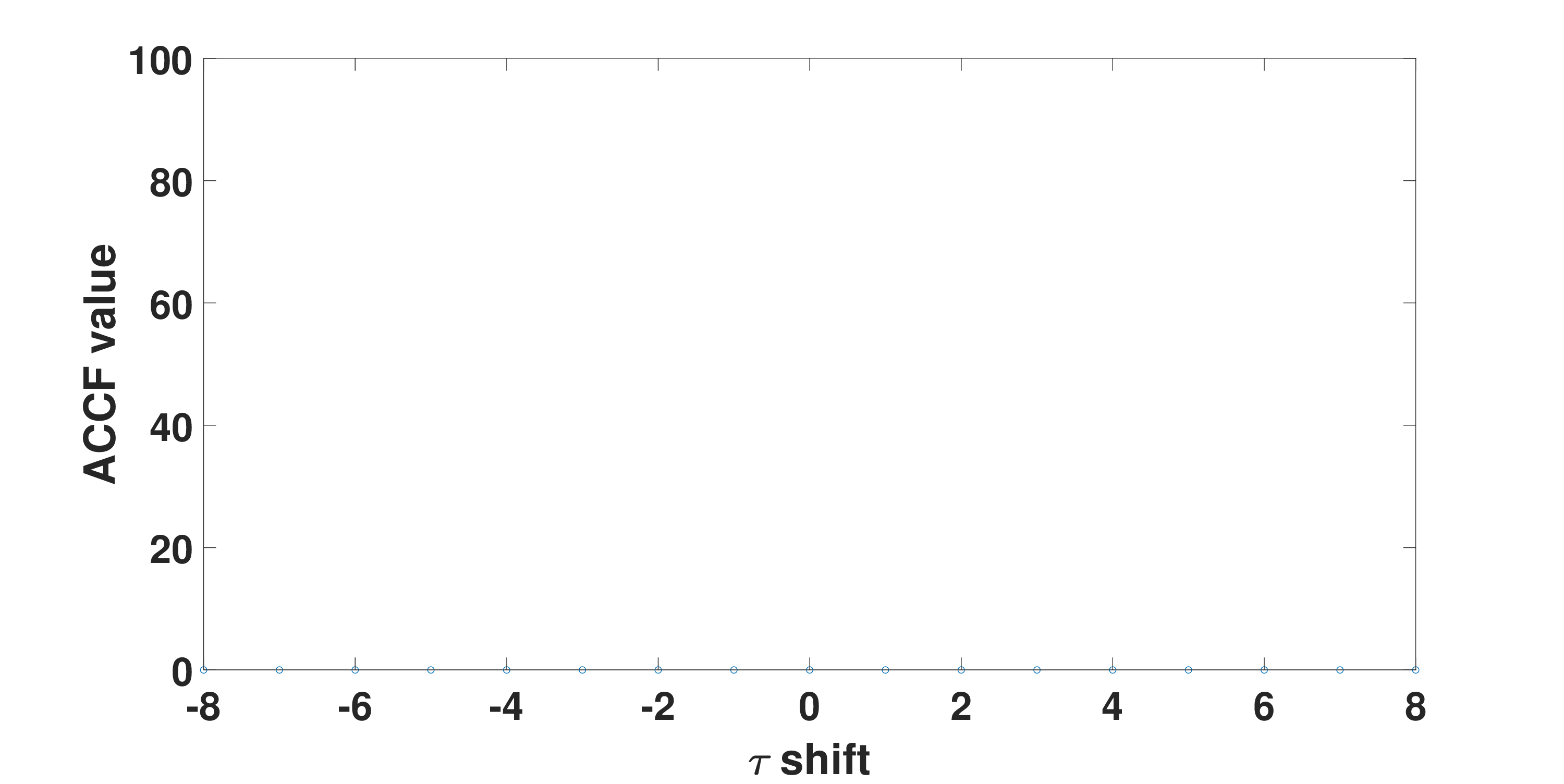}
\caption{Cross-correlation result of any set of array from $\mathcal{S}$}\label{fig4}
\end{figure}
\section{Constructions of ZCCS}
Let $p_{1},p_{2},\ldots,p_{l}$ be $l$ primes and $q'=q\prod_{i=1}^{l}p_{i}$. Define the function $G_{k,l}:\{0,1,\ldots,q'-1\}\rightarrow \mathbf{C}$ by
$$G^{\mathbf{c}}_{k,l}(i')=S_{k,l}(i)\prod_{m=1}^{l}\omega_{p_{m}}^{c_{m}i_{m}},$$ where 
$$i'=i+i_{1}q+i_{2}p_{1}q+\ldots+i_{l}p_{l-1}\ldots p_{1}q,$$ and $0\leq i\leq q-1$ and $0\leq i_{t}\leq p_{t}-1,1\leq t\leq l$.
We define the sequence corresponding to $G^{\mathbf{c}}_{k,l}$ as $\psi(G^{\mathbf{c}}_{k,l})$ where
$$\psi(G^{\mathbf{c}}_{k,l})=(G^{\mathbf{c}}_{k,l}(0),G^{\mathbf{c}}_{k,l}(1),\ldots,G^{\mathbf{c}}_{k,l}(q'-1)).$$
\begin{theorem}
\label{RK}
    Consider the ordered set of sequences or code set as
    $$\psi(G^{\mathbf{c}}_{k})=\{\psi(G^{\mathbf{c}}_{k,l}):0\leq l \leq q-1 \}.$$
    Define the code set
    $$T=\{\psi(G^{\mathbf{c}}_{k}):0\leq k\leq q-1,0\leq c_{i}\leq p_{i}-1\}.$$ Then the code set $T$ forms an optimal ZCCS with parameter $(nq,q,q,nq)$ where $n=\prod_{i=1}^{l}p_{i}$.
\end{theorem}
\begin{proof}
    For $0< \tau \leq q-1$, the ACCF between $\psi(G^{\mathbf{c}}_{k})$ and $\psi(G^{\mathbf{c'}}_{k'})$ can be derived as
    \begin{equation}
    \label{R}
    \begin{split}
&\Phi(\psi(G^{\mathbf{c}}_{k}),\psi(G^{\mathbf{c'}}_{k'}))(\tau)\\
&=\Phi\left(\psi(S_{k}),\psi(S_{k'})\right)(\tau)\sum_{\alpha=0}^{\prod_{i=1}^{l}p_{i}-1}\prod_{t=1}^{l}\omega_{p_{t}}^{(c_{t}-c'_{t})(\alpha_{t})}\\
&+\Phi\left(\psi(S_{k}),\psi(S_{k'})\right)(\tau-q)\sum_{\alpha=0}^{\prod_{i=1}^{l}p_{i}-2}\prod_{t=1}^{l}\omega_{p_{t}}^{(c_{t}(\alpha_{t})-c'_{t}(\alpha'_{t}))}
    \end{split}
    \end{equation}
    where $$\alpha = \alpha_{1}+\sum_{i=2}^{l} \left( \prod_{j=1}^{i-1} p_j \right) \alpha_i,~\alpha+1 = \alpha'_{1}+\sum_{i=2}^{l} \left( \prod_{j=1}^{i-1} p_j \right) \alpha'_i,$$ and $0\leq \alpha_{t},\alpha'_{t}\leq p_{t}-1 $.
    From Theorem 1 and (\ref{R}) we have 
\begin{equation}
\label{Hari}
    \Phi(\psi(G^{\mathbf{c}}_{k}),\psi(G^{\mathbf{c'}}_{k'}))(\tau)=0.
\end{equation}
For $\tau=0$, the ACCF between $\psi(G^{\mathbf{c}}_{k})$ and $\psi(G^{\mathbf{c'}}_{k'})$ can be derived as
\begin{equation}
\label{vrinda}
\begin{split}
    &\Phi(\psi(G^{\mathbf{c}}_{k}),\psi(G^{\mathbf{c'}}_{k'}))(0)\\
    &=\Phi\left(\psi(S_{k}),\psi(S_{k'})\right)(0)\sum_{\alpha=0}^{\prod_{i=1}^{l}p_{i}-1}\prod_{t=1}^{l}\omega_{p_{t}}^{(c_{t}-c'_{t})(\alpha_{t})}\\
    &=\Phi\left(\psi(S_{k}),\psi(S_{k'})\right)(0)\prod_{t=1}^{l}\left(\sum_{\alpha_{t}=0}^{p_{t}-1}\omega_{p_{t}}^{(c_{t}-c'_{t})(\alpha_{t})}\right).
\end{split}
\end{equation}
To derive the ACCF between $\psi(G^{\mathbf{c}}_{k})$ and $\psi(G^{\mathbf{c'}}_{k'})$ we have the following cases
\begin{enumerate}
    \item Case 1 $(k\neq k')$: In this case we have
    \begin{equation}
    \label{ban}
    \Phi\left(\psi(S_{k}),\psi(S_{k'})\right)(0)=0,    \end{equation}
     by \textit{theorem} 1. Therefore, from (\ref{vrinda}) and (\ref{ban})  we have $$\Phi(\psi(G^{\mathbf{c}}_{k}),\psi(G^{\mathbf{c'}}_{k'}))(0)=0.$$
    \item Case 2 $(\mathbf{c}\neq\mathbf{c'})$: Then at least one of $c_{t}$ and $c'_{t}$ deffers. In this case, we have
    \begin{equation}
    \label{bharat}\sum_{\alpha_{t}=0}^{p_{t}-1}\omega_{p_{t}}^{(c_{t}-c'_{t})(\alpha_{t})}=0. 
    \end{equation}
     Therefore from (\ref{bharat}) and (\ref{vrinda}) we have $$\Phi(\psi(G^{\mathbf{c}}_{k}),\psi(G^{\mathbf{c'}}_{k'}))(0)=0.$$
    \item Case 3 $(\mathbf{c}=\mathbf{c'}),k=k'$: In this case we have
    $$\Phi(\psi(G^{\mathbf{c}}_{k}),\psi(G^{\mathbf{c'}}_{k'}))(0)=q^{2}\left(\prod_{i=1}^{l}p_{i}\right).$$ Hence by (\ref{Hari}), case 1, case 2 and case 3 we conclude that the set $T$ forms optimal ZCCS with parameter $(nq,q,q,nq)$. 
\end{enumerate}
\end{proof}
\begin{example}
\label{K2}
    Let  $p_{1}=2,q=9$ and $S_{k,l}:\{0,1,\ldots,8\}\rightarrow $
    $\mathbf{C}$ be the function as defined in example 1. We define the function $G_{k,l}:\{0,1,\ldots,17\}\rightarrow \mathbf{C}$ by
$$G^{\mathbf{c}}_{k,l}(i')=S_{k,l}(i)\omega_{p_{1}}^{c_{1}i_{1}},$$ where 
$i'=i+i_{1}9$ and $0\leq i\leq 8$ and $0\leq i_{1}\leq 1$. Consider $$\psi(G^{\mathbf{c}}_{k})=\{\psi(G^{\mathbf{c}}_{k,l}):0\leq l \leq 8 \},$$ where $$\psi(G^{\mathbf{c}}_{k,l})=(G^{\mathbf{c}}_{k,l}(0),G^{\mathbf{c}}_{k,l}(1),\ldots,G^{\mathbf{c}}_{k,l}(17)).$$ Hence by theorem \ref{RK} the code set $$T=\{\psi(G^{\mathbf{c}}_{k}):0\leq k\leq 8,0\leq c_{1}\leq 1\},$$ forms optimal ZCCS with parameter $(18,9,9,18)$. Fig. \ref{fig5} and Fig. \ref{fig6} shows the auto and cross-correlation result from any codes from $T$ in example \ref{K2}.  
\end{example}
\begin{figure}
\centering
\includegraphics[width=0.5\textwidth]{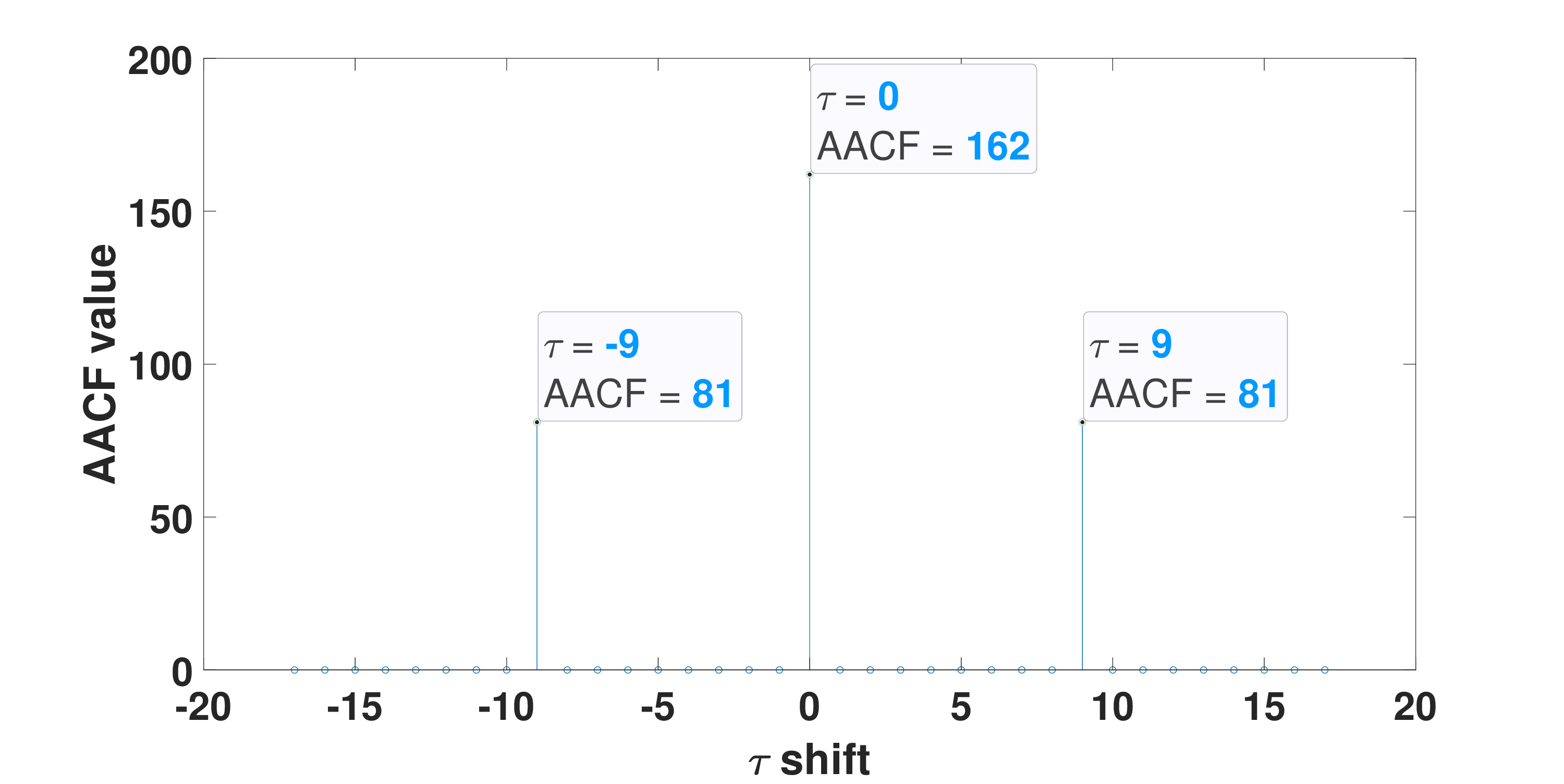}
\caption{Auto-correlation result of any set of array from $T$}\label{fig5}
\end{figure}
\begin{figure}
\centering
\includegraphics[width=0.5\textwidth]{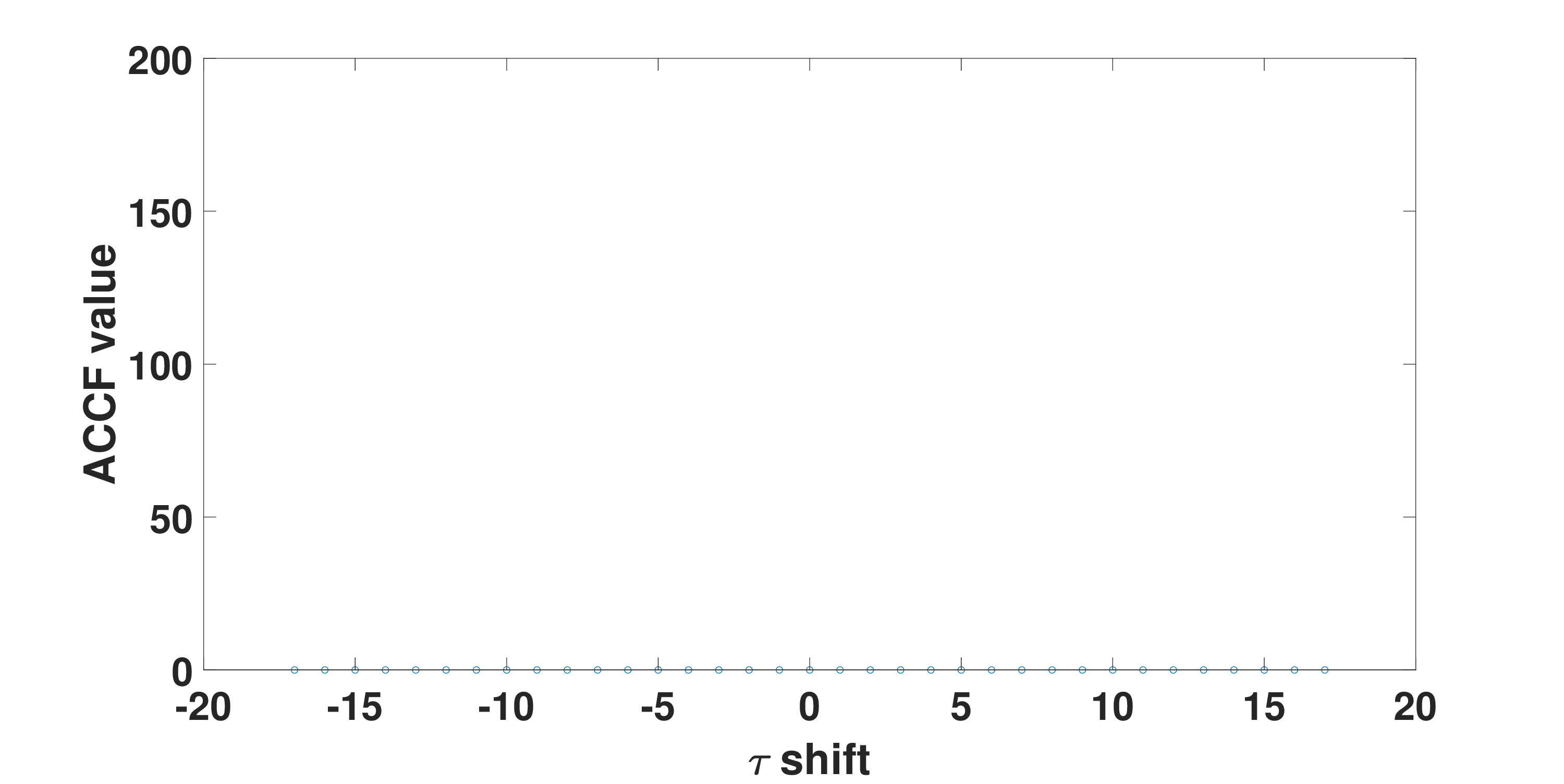}
\caption{Cross-correlation result of any set of array from $T$}\label{fig6}
\end{figure}
\begin{table}[htbp]
\caption{Comparison with Existing Optimal ZCCS}
\label{tab3}
\centering
\begin{tabular}{|l|l|c|c|c|}
\hline
Source & Based On & Length & Users & ZCZ \\
\hline
\cite{shen2023new} & Boolen function & $q^{m}$ & $q^{v+1}$ & $q$ \\
\hline
\cite{xie2021constructions} & Boolen function & $3.2^{m}$ & $2^{k+1}$ & $2^{m}$\\
\hline
\cite{ghosh2022direct} & Boolen function & $\prod_{i=1}^{l}p_{i}2^m$ & $(\prod_{i=1}^{l}p_{i}2^{n+1})$ & $2^{m}$ \\
\hline
\cite{ghosh2023construction} & Boolen function & $R2^m$ & $R2^{k+1}$ & $2^{m}$ \\
\hline
Proposed & Galois Field & $np^{r}$ & $np^{r}$ & $p^{r}$\\
\hline
\end{tabular}
\end{table}
\begin{remark}
Our proposed construction method generates ZCCS of length \(np^{r}\), where \(n=\prod_{i=1}^{l}p_{i}\) and each \(p_i\) is a prime number. By choosing specific values such as \(l=1\), \(p=2\), and \(p_{1}=3\), our construction yields ZCCS of length \(3.2^{r}\). This indicates that the ZCCS lengths described in \cite{xie2021constructions} can be considered as a special case within our framework. Furthermore, with an appropriate selection of \(n\), \(p\), and \(r\), it can be observed that the ZCCS lengths mentioned in \cite{ghosh2022direct,ghosh2023construction} also emerge as special cases of our construction.
\end{remark}

\section{Conclusion}

This work proposes new optimal ZCCS, which are useful in wireless communication especially in MC-CDMA systems to accommodate a larger number of users. Our approach diverges from the conventional generalized Boolean function-based constructions. By utilizing additive characters over finite Galois fields, we have developed CCC and ZCCS.  The CCC in our proposed construction has a length of the form $p^{r}$. For the ZCCS, both the length and sizes are characterized by $np^{r}$, where $n,r$ are any positive numbers, and $p$  is a prime number. This particular length and user capacity of the proposed ZCCS are previously unreported in the literature. Additionally, several existing lengths of ZCCS are encompassed as special cases within our framework.

\IEEEtriggeratref{14}
\bibliographystyle{IEEEtran}
\bibliography{bibliography.bib}

\end{document}